\documentclass[twoside]{article}

\usepackage[accepted]{aistats2025}

\usepackage{amsthm,amssymb}
\usepackage{mathtools}
\usepackage{graphicx}
\usepackage{subfig}
\usepackage{hyperref}
\usepackage{wrapfig}
\usepackage{booktabs}
\usepackage{thm-restate}
\usepackage{dsfont}
\usepackage{stmaryrd}
\usepackage{textcomp}

\theoremstyle{plain}
\newtheorem{theorem}{Theorem}[section]
\newtheorem{proposition}[theorem]{Proposition}
\newtheorem{lemma}[theorem]{Lemma}

\newtheorem{claim}[theorem]{Claim}
\theoremstyle{definition}
\newtheorem{definition}[theorem]{Definition}
\newtheorem{observation}[theorem]{Observation}

\theoremstyle{remark}
\newtheorem{remark}[theorem]{Remark}

\hypersetup{
	colorlinks,
	linkcolor={red!50!black},
	citecolor={blue!50!black},
	urlcolor={blue!80!black}
}
\usepackage{thmtools}      %
\usepackage{colortbl}         %
\usepackage{enumitem}
\usepackage{amsfonts}       %
\usepackage[linesnumbered,algo2e,ruled,noend]{algorithm2e}

\usepackage{caption}
\usepackage{tikz}
\usetikzlibrary{arrows,shapes,positioning}
\usepackage{pgfplots}
\usepackage{tikzscale}
\usepgfplotslibrary{fillbetween}

\declaretheorem[name=Corollary]{cor}

\usepackage[capitalize,noabbrev]{cleveref}
\usepackage{xspace}

\newcommand{\sts}{\ensuremath{\mathsf{STS}}\xspace}
\newcommand{\cmsgen}{\ensuremath{\mathsf{CMSGen}}\xspace}

\newcommand{\SubVsSubMain}{\ensuremath{\mathsf{DistEstimate}}\xspace}
\newcommand{\SubVsSub}{\ensuremath{\mathsf{DistEstimateCore}}\xspace}
\newcommand{\SubToAeval}{\ensuremath{\mathsf{SubToEval}}\xspace}

\newcommand{\samp}{\ensuremath{\mathsf{SAMP}}\xspace}

\newcommand{\cond}{\ensuremath{\mathsf{COND}}\xspace}

\newcommand{\condmar}{\ensuremath{\mathsf{CM}}\xspace}
\newcommand{\subcond}{\ensuremath{\mathsf{SUBCOND}}\xspace}

\newcommand{\expect}{\mathbb{E}}
\newcommand{\variance}{\mathbb{V}}
\newcommand{\nb}{\mathsf{NB}}
\newcommand{\ba}{\{0,1\}}

\newcommand{\bn}{\{0,1\}^n}

\usepackage[round]{natbib}

\begin{document}

\twocolumn[

\aistatstitle{Distance Estimation for High-Dimensional Discrete Distributions}

\aistatsauthor{Gunjan Kumar \textcircled{r} \And Kuldeep S. Meel \textcircled{r} \And Yash Pote}

\aistatsaddress{ IIT-Kanpur \And  Georgia Institute of Technology \And National University of Singapore\\ CREATE } ]

\begin{abstract}
Given two distributions $\mathcal{P}$ and $\mathcal{Q}$ over a high-dimensional domain $\{0,1\}^n$, and a parameter $\varepsilon$, the goal of distance estimation is to determine the statistical distance between $\mathcal{P}$ and $\mathcal{Q}$, up to an additive tolerance $\pm \varepsilon$. Since exponential lower bounds (in $n$) are known for the problem in the standard sampling model, research has focused on richer query models where one can draw conditional samples. This paper presents the first polynomial query distance estimator in the conditional sampling model ($\mathsf{COND}$). 

We base our algorithm on the relatively weaker \textit{subcube conditional} sampling ($\mathsf{SUBCOND}$) oracle, which draws samples from the distribution conditioned on some of the dimensions. $\mathsf{SUBCOND}$ is a promising model for widespread practical use because it captures the natural behavior of discrete samplers. Our algorithm makes $\tilde{\mathcal{O}}(n^3/\varepsilon^5)$ queries to $\mathsf{SUBCOND}$.
\end{abstract}

\section{INTRODUCTION}\label{sec:intro}

Given two discrete distributions $\mathcal{P}$ and $\mathcal{Q}$ over $\{0,1\}^n$, the total variation (TV) distance between $\mathcal{P}$ and $\mathcal{Q}$, denoted by $d_{TV}(\mathcal{P},\mathcal{Q})$, is defined as:
\begin{align*}
    d_{TV}(\mathcal{P},\mathcal{Q}) = \frac{1}{2}\sum_{\sigma \in \{0,1\}^n} |\mathcal{P}(\sigma) - \mathcal{Q}(\sigma)|
\end{align*}
In this paper, we are interested in the computation of $(\varepsilon,\delta)$-approximation of $d_{TV}(\mathcal{P},\mathcal{Q})$: i.e., we would like to compute an estimate $\mathtt{est}$ such that $\Pr[ d_{TV}(\mathcal{P},\mathcal{Q}) -\varepsilon \leq \mathtt{est} \leq d_{TV}(\mathcal{P},\mathcal{Q}) + \varepsilon] \geq 1-\delta$.
TV distance is a fundamental notion in probability and finds applications in diverse domains of computer science such as generative models~\citep{GPMX+,JKHZ23},  MCMC algorithms~\citep{ADDJ+03,BDX04,BGJ11}, and probabilistic programming~\citep{ABH21,PM22}. 

Theoretical investigations into the problem of TV distance computation have revealed the intractability of exact computation: In particular, the problem is \#P-hard even when $\mathcal{P}$ and $\mathcal{Q}$ are represented as product distributions~\citep{BGMM+23}. As a consequence, the focus has been on designing approximation techniques. Randomized polynomial-time approximation schemes are known for some classes of distributions when $\mathcal{P}$ and $\mathcal{Q}$ are specified explicitly. An example is Bayesian networks with bounded treewidth~\citep{BGMM+23b}. Not every practical application allows explicit representation of probability distributions, and often, the output of some underlying process defines probability distributions. Accordingly, the field of distribution testing is concerned with the design of algorithmic techniques for different models of access to the underlying processes. Furthermore, in addition to the classical notion of time complexity, we are also concerned with the {\em query complexity}: how many queries do we make to a given access model?   

The earliest investigations focused on the classical model of access where one is only allowed to access samples from $\mathcal{P}$ and $\mathcal{Q}$ \citep{PL08,VV11a}; however, a lower bound of  $\Omega(2^n/n)$~\citep{VV10,VV11a}  restricts the applicability of these estimators in practical scenarios. This motivates the need to focus on more powerful models. In this work, we will focus on the {\subcond} access model owing to its ability to capture the behavior of probabilistic processes in diverse settings~\citep{JVV86,CMN99,ZCLH18}. For example,  \subcond access perfectly models autoregressive sampling as employed in state-of-the-art LLMs and image models~\citep{OKK16, KESO16}.

Formally, the {\subcond} oracle for a distribution $\mathcal{P}$ takes in a query string $\rho \in \{0,1,*\}^n$,  constructs the conditioning set $S_{\rho} = \{\sigma \in \bn |(\rho_i = *)\; \vee (\rho_i = \sigma_i)\}$ and  returns $\sigma \in S_\rho$ with probability $\frac{\mathcal{P}(\sigma)}{\sum_{ \pi \in S_{\rho} } \mathcal{P}(\pi)}$. It is worth remarking that while we use the name {\subcond} to be consistent with recent literature ~\citep{BC18}, there have been algorithmic frameworks since the late 1980s that have relied on the underlying query model~\citep{JVV86}.

The starting point of our investigation is the observation that, on the one hand, practical applications of distance estimation rely on heuristic methods and hence don't provide any guarantees. On the other hand, no known algorithm, even when given access to the {\subcond} oracle, makes less than $O(2^n/n)$ queries. The primary contribution of our work is to address the mentioned gap: we design the first algorithm that computes $(\varepsilon,\delta)$-approximation of TV distance and makes only polynomially many queries to {\subcond} oracle. Formally, 
  
\begin{restatable}{theorem}{scondvsscond}\label{thm:scondvsscond}
Given two distributions $\mathcal{P}$ and $\mathcal{Q}$ over $\bn$, along with parameters $\varepsilon \in(0,1)$, and $\delta \in (0,1)$, the algorithm $\SubVsSubMain(\mathcal{P},\mathcal{Q},\varepsilon,\delta)$  returns estimate $\kappa$ such that $$ \Pr[  \kappa \in (d_{TV}(\mathcal{P},\mathcal{Q}) \pm \varepsilon)  ] \geq 1-\delta$$
\SubVsSubMain makes $\tilde{\mathcal{O}}\left(n^3\log(1/\delta)/\varepsilon^4 \right)$ queries to the {\subcond} oracle.
 \end{restatable}

We now provide a high-level overview of {\SubVsSubMain}: From the fact that,
\begin{align*}
    d_{TV}(\mathcal{P},\mathcal{Q}) =  \underset{\sigma \sim \mathcal{Q}}{\mathbb{E}}\left[ \max \left(1-\frac{\mathcal{P}(\sigma)}{\mathcal{Q}(\sigma)},0\right)\right]
\end{align*} 
we can use the standard approach of sampling $\sigma$ from $\mathcal{Q}$, estimating $\mathcal{P}(\sigma)$ and $\mathcal{Q}(\sigma)$ up to some multiplicative factor, and then setting the value of the random variable to be $\max(1-\mathcal{P}(\sigma)/\mathcal{Q}(\sigma),0)$. This approach requires a constant number of samples from $\mathcal{Q}$ to compute an approximation of $d_{TV}(\mathcal{P},\mathcal{Q})$. The main issue is that it is not possible to approximate the value of $\mathcal{Q}(\sigma)$ for arbitrary $\sigma$ with only polynomially many queries to $\subcond$ since $\mathcal{Q}(\sigma)$ can be arbitrarily small and the query complexity scales inversely with $\mathcal{Q}(\sigma)$. The key technical contribution lies in showing that using polynomially many {\subcond} oracle calls, we can still compute estimates for $\mathcal{P}(\sigma)$ and $\mathcal{Q}(\sigma)$ at sufficiently many points to find a theoretically guaranteed estimate.

We are interested in designing distance estimation techniques for the {\subcond} model because it effectively captures the behavior of probabilistic processes in practice. Towards this goal, we compute the precise number of queries one would need to the test, and we find that  {\SubVsSubMain} offers a $10^7$ factor speedup on problems of dimensionality $n=70$, for which the baseline sample-based estimator would require $\simeq 10^{18}$ queries -- a prohibitively large number. The result is presented in the Figure~\ref{fig:valiant}.
Therefore, we demonstrate the application of {\SubVsSubMain} in a real-world setting. Sampling from discrete domains such as $\{0,1\}^n$ under combinatorial constraints is a challenging problem; therefore, several heuristic-based samplers have been proposed over the years. We can view a sampler as a probabilistic process, and consequently, one is interested in measuring how far the distribution of a given sampler is from the ideal distribution. Our experiments focus on combinatorial samplers, and \subcond is particularly well suited for this problem. We use a prototype of {\SubVsSubMain} to evaluate the quality of two samplers for different benchmarks. Our empirical evaluation demonstrates the promise of scalability: in particular, {\SubVsSubMain} offers a $10^7$ factor speedup on problems of dimensionality $n=70$.

\begin{figure}
    \centering
    \includegraphics[scale=0.33]{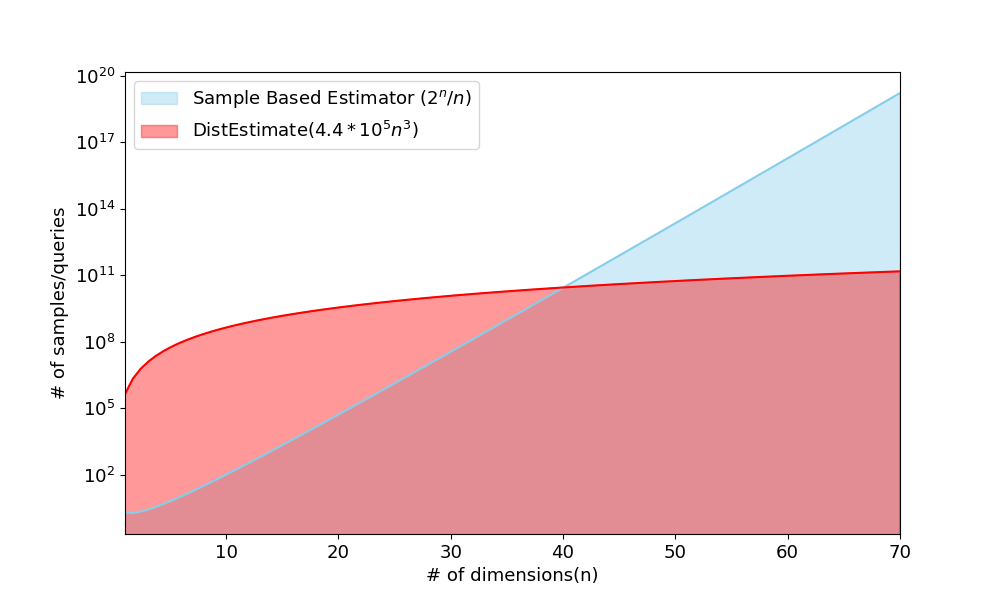}
    \caption{A plot comparing the sample/query complexity of the baseline non-conditional estimator vs. our estimator $\SubVsSubMain$ as a function of the number of dimensions $n$, for $\varepsilon=0.3$. Note that the vertical axis is in the log scale.}
    \label{fig:valiant}
\end{figure}

\paragraph{Organization}
We start with a short background on related threads of investigation in Section~\ref{sec:relatedwork}. Then in Section~\ref{sec:notation} we define the notation we use in most of the paper. We present the paper's main contribution, the estimator $\SubVsSubMain$, along with its proof of correctness in Section~\ref{sec:subcondvssubcond}.  In Section~\ref{sec:experiments}, we present the result of the evaluation of our implementation of \SubVsSubMain. Finally, we conclude in Section~\ref{sec:conclusion} and discuss some open problems. In the interest of exposition, we defer some proofs to the Appendix.

\section{RELATED WORK}\label{sec:relatedwork}

Distance estimation is one of the many problems in the broader area of distribution testing. Apart from estimation, there is extensive literature on the problems of identity and equivalence testing. The problem of identity testing involves returning $\mathsf{Accept}$ if $d_{TV}(\mathcal{P},\mathcal{P}^*) = 0$ and returning $\mathsf{Reject}$ if $d_{TV}(\mathcal{P},\mathcal{P}^*) >\varepsilon$, where $\mathcal{P}$ is an unknown distribution and $\mathcal{P}^*$ is known, i.e. you have a full description of $\mathcal{P}^*$. Equivalence testing is the generalization of identity testing. It is the problem of deciding between $d_{TV}(\mathcal{P},\mathcal{Q}) = 0$ and $d_{TV}(\mathcal{P},\mathcal{Q}) >\varepsilon$ where both $\mathcal{P}$ and $\mathcal{Q}$ are unknown. It is worth emphasizing that for both identity and equivalence testing problems, any answer from the tester ($\mathsf{Accept}$ or $\mathsf{Reject}$) is considered valid if  $ 0 < d_{TV}(\mathcal{P},\mathcal{Q}) \le \varepsilon$. Provided only sample access, the sample complexity of identity testing is $\Theta\left(2^{n/2}/\varepsilon^2\right)$ \citep{PL08,VV17} and of equivalence testing is $\max(2^{2n/3}\varepsilon^{-4/3},2^{n/2}\varepsilon^{-2})$ \citet{CDVV14,VV17}. While testing is of theoretical interest, its practical application faces significant limitations primarily because testers must accept only when two given distributions are identical. In real-world scenarios, distributions are rarely identical but often exhibit close similarity. Consequently, a simplistic tester that consistently returns $\mathsf{Reject}$  can meet the specifications. A more rigorous definition of a tester is required to address this limitation, including estimating the distance between the two distributions. Unfortunately, this introduces a considerable challenge. \cite{VV11a} demonstrate that in the classical sampling model, the necessary number of queries increases to $2^n/n$, a significant jump from the previous $2^{2n/3}$.

To sidestep the exponential lower bounds on testing, the conditional sampling model, or \cond, was introduced independently by \citeauthor{CFGM16} and \citeauthor{CRS15}, and has been successfully applied to various problems, including identity and equivalence testing. In this model, the sample complexity of identity testing is $\Theta(\varepsilon^{-2})$ (independent of $n$), while for equivalence testing the best-known upper and lower bounds are $O((\log n)/\varepsilon^{5})$ \citep{FJOPS15}, and $\Omega(\sqrt{\log n})$ \citep{ACK14} respectively. A survey by \citet{C20} provides a detailed view of testing and related problems in various sampling models.

Our work investigates the distance estimation problem using the \subcond model, a restriction of $\mathsf{COND}$. Unlike \cond, which allows conditioning on arbitrary sets, the \subcond model allows conditioning only on sets that are subcubes of the domain. While \cond significantly improves the sample complexity, it is not easily implementable in practice, as arbitrary subsets are not efficiently represented and sampled from.
With a view towards plausible conditional models, \citet{CRS15,BC18} came up with the $\mathsf{SUBCOND}$ model, which is particularly suited to the Boolean hypercube $\bn$. \citet{CCKLW21} used the $\mathsf{SUBCOND}$ model to construct a nearly-optimal $\Theta(\sqrt{n})$ uniformity testing algorithm for $\{0,1\}^n$, demonstrating its natural applicability for high-dimensional distributions. Then \citet{CJLW21} used $\mathsf{SUBCOND}$ to study the problems of learning and testing junta distributions supported on $\{0,1\}^n$.  \citet{BC18} developed a test for equivalence in the $\mathsf{SUBCOND}$ model, with query complexity of $O(n^2/\varepsilon^2)$. However, before this work, there was no distance estimation algorithm in the $\mathsf{SUBCOND}$ oracle model, and indeed even in the general \cond model. 

\paragraph{Lower Bound}\label{sec:lowerbound}
The problem of testing with \subcond access has a query complexity lower bound of $\Omega(n/\log(n))$ as a direct consequence of Theorem 11 of ~\citet{CDKS20}. For completeness, we formally prove the lower bound  in Appendix~\ref{sec:lboundforsvs}.

\section{NOTATIONS AND PRELIMINARIES}\label{sec:notation}

We will focus on probability distributions over $\bn $. For any distribution $\mathcal{D}$ on $\bn$ and an element $\sigma \in \bn$, $\mathcal{D}(\sigma)$ is the  probability of $\sigma$ in distribution $\mathcal{D}$. 
Further, $\sigma \sim \mathcal{D}$ represents that $\sigma$ is sampled from $\mathcal{D}$. 
The total variation ($\mathrm{TV}$) distance of two probability distributions $\mathcal{P}$ and $\mathcal{Q}$ is defined as: $ d_{TV}(\mathcal{P},\mathcal{Q})= \frac{1}{2}\sum_{\sigma \in \bn} |\mathcal{P}(\sigma)-\mathcal{Q}(\sigma)|$. For a random variable $v$, the expectation is denoted as $\expect[v]$, and the variance as $\variance[v]$. 

For clarity of exposition, we will hide the use of the ceiling operator $\lceil x \rceil$ wherever integral values are required, such as the number of samples or the number of iterations of a loop. We use $[n]$ to represent the set $\{1,2\ldots,n\}$.

Consider a discrete r.v. that takes the value $v$ with probability $p$. The count of trials required to observe $k$ instances of $v$ follows a \textit{negative binomial} distribution, denoted as $\nb(k,p)$. 
The expected value $\expect[\nb(k,p)]$ is $k/p$, and its variance $\variance[\nb(k,p)]$ is $k(1-p)/p^2$. We state the Chernoff and Chebyshev concentration bounds in  Appendix~\ref{sec:appendix:useful} for completeness. Further, we also make use of the following tail bound for negative binomials:
\begin{proposition}[\cite{B11}]\label{prop:negbinconc}
	For $\gamma > 1$, $\Pr[\nb(k,p) > \gamma\expect[\nb(k,p)] ] \leq \exp\left(-\frac{\gamma k(1-1/\gamma)^2}{2}\right)$
\end{proposition}

If $\sigma$ is a string of length $n>0$, then $\sigma_i$ denotes the $i^{th}$ element of $\sigma$, and for $1\leq j\leq n$, $\sigma_{<i}$ denotes the substring of $\sigma$ from $1$ to $i-1$,   $\sigma_{<i} = \sigma_1\cdots\sigma_{i-1}$; similarly $\sigma_{\leq i} = \sigma_1\cdots\sigma_{i}$, and $\sigma_{<1}$ denotes the empty (length 0) string, also denoted as $\bot$. 

For any distribution $\mathcal{D}$ and string $\rho$, such that $0\leq |\rho| < n$, the distribution  $\mathcal{D}_{\rho}$ denotes the marginal distribution of $\mathcal{D}$  in the ${|\rho|+1}^{th}$ dimension, conditioned on the string $\rho$, i.e., 
$\mathcal{D}_\rho(b) = \frac{\Pr_{\sigma \sim \mathcal{D}}[ (\sigma_{|\rho|+1} = b )\wedge (\sigma_{\leq|\rho|} = \rho )]}{\Pr_{\sigma \sim \mathcal{D}}[ \sigma_{\leq |\rho|} = \rho ]} $.

\begin{definition}
	A sampling oracle $\samp(\mathcal{D})$  takes in a distribution $\mathcal{D}$, and returns a sample $\sigma  \in \{0,1\}^n $ such that $	\Pr\left[\samp(\mathcal{D}) = \sigma \right] =	\mathcal{D}(\sigma)$.
\end{definition}

\begin{definition}
	A subcube conditioning oracle $\subcond(\mathcal{D},\rho)$ takes in a distribution $\mathcal{D}$, and a query string $\rho$ with $0 \le |\rho| < n$, and returns a sample $\sigma  \in \{0,1\}^n $ such that $\Pr\left[\subcond(\mathcal{D},\rho) = \sigma \right] = 1_{(\sigma_{\leq |\rho|} = \rho)}	\prod_{i=|\rho|+1}^{n}\mathcal{D}_{\sigma_{< i}}(\sigma_{i})$.
\end{definition}

\begin{definition}
	A conditional marginal oracle $\condmar(\mathcal{D},\rho)$ takes in a distribution $\mathcal{D}$, and a query string $\rho$ with $0 \le |\rho| < n$, and returns a sample $b  \in \{0,1\} $ such that $	\Pr\left[\condmar(\mathcal{D},\rho) = b \right] =	\mathcal{D}_{\rho}(b)$.
\end{definition}

Note that the chain rule implies that $\subcond(\mathcal{D},\bot)$ is the same as $\samp(\mathcal{D})$.

\subsection{Distance Approximation}
We adapt the distance approximation algorithm of~\citet{BGMV20}, that takes as input  two distributions $\mathcal{P}$ and $\mathcal{Q}$, and provides an $(\eta,\delta)$ estimate of $d_{TV}(\mathcal{P}, \mathcal{Q})$.  The proof is deferred to Appendix~\ref{sec:appendix:evalvseval}. 
\begin{restatable}{lemma}{evalvseval}{(Theorem 3.1 in~\citep{BGMV20})}\label{lem:evalvseval}	
	For distributions $\mathcal{P}$ and $\mathcal{Q}$ over $\bn$, and $\sigma \in \bn$,  let $p_\sigma$ and $q_\sigma$ be functions such that $p_\sigma \in (1\pm \eta)\mathcal{P}(\sigma)$, and ${q_\sigma} \in (1\pm\eta)\mathcal{Q}(\sigma)$.  
	Given a set of samples $S$ from  $\mathcal{Q}$, and $\eta \in (0,1)$  
	along with the $p_\sigma$ and $q_\sigma$ for each $\sigma \in S$, let $\mathtt{est} = \frac{1}{|S|}\sum_{i \in S} 1_{q_\sigma > p_\sigma} \left(1-\frac{p_\sigma}{q_\sigma}\right)$. 
	\begin{align*}
	&\Pr\left[  \mathtt{est} \not \in  \left( d_{TV}(\mathcal{P},\mathcal{Q}) \pm \frac{3\eta}{1-\eta}\right) \right]  \\
	&\leq 2\exp\left(-2|S|\left(\frac{\eta }{1-\eta}\right)^2\right)
	\end{align*}
\end{restatable}

\subsection{Taming Distributions}\label{sec:taming}
Given a distribution $\mathcal{D}$, we will define and construct a new distribution $\mathcal{D'}$ that has desirable properties critical for \SubVsSubMain. 
\begin{definition}\label{def:tame}
	A distribution $\mathcal{D'}$ is $\theta$-tamed,  if $$\forall \sigma \in \bn,\forall \ell \in [n] \quad \mathcal{D'}_{\sigma_{<\ell}}(\sigma_\ell) \in [\theta, 1-\theta]$$
\end{definition}
\begin{definition}
	For a given distribution $\mathcal{D}$, and parameter $\theta \in [0,1/n)$,  distribution  $\mathcal{D'}$ is the $\theta$-tamed sibling of $\mathcal{D}$, if $\mathcal{D'}$ is $\theta$-tamed and  $d_{TV}(\mathcal{D},\mathcal{D'}) \leq \theta n$.
\end{definition}

Henceforth, we will use $\mathcal{D'}$ as shorthand to refer to the $\theta$-tamed sibling of $\mathcal{D}$ and omit mentioning $\theta$ whenever $\theta$ is evident from the context. We will now show in the following lemma that given \subcond query access to distribution $\mathcal{D}$,  \condmar, and \samp access to $\mathcal{D'}$ can be simulated efficiently. We defer the proof to Appendix~\ref{sec:appendix:tame}.

\begin{restatable}{lemma}{tame}\label{lem:tame} 
	Given a distribution $\mathcal{D}$ and parameter $\theta \in [0,1/n)$, every \condmar query to $\mathcal{D'}$ can be simulated by making one  \subcond query to $\mathcal{D}$, and every \samp query to $\mathcal{D'}$ can be simulated by making $n$  \subcond queries to $\mathcal{D}$.
\end{restatable}

\section{\SubVsSubMain: A DISTANCE ESTIMATION ALGORITHM}\label{sec:subcondvssubcond}
We now present the pseudocode of our algorithm \SubVsSubMain, and the \SubToAeval and \SubVsSub subroutines. 
The following subsection will provide a high-level overview of all our algorithms and formal analysis.
	{\setlength{\algomargin}{1em}
		{\SetKwComment{Comment}{$\triangleright$ }{}
			\begin{algorithm2e}[ht]
				\DontPrintSemicolon
				\caption{$\SubVsSubMain(\mathcal{P},\mathcal{Q},\varepsilon,\delta)$}\label{alg:subvssubmain}    
				All{$j =1$ {\bfseries to} $4.5\log(2/\delta)$}{
					$r_j \gets \SubVsSub(\mathcal{P},\mathcal{Q},\varepsilon)$
				}			
				$ \kappa \gets  \mathtt{Median}_{j} (r_j)$\\
				\Return  $\kappa$
			\end{algorithm2e}
			\begin{algorithm2e}[ht]
				\DontPrintSemicolon
				\caption{$\SubVsSub(\mathcal{P},\mathcal{Q},\varepsilon)$}\label{alg:subvssub}
				\Comment{$\mathcal{P'}$ and $\mathcal{Q'}$ are $\varepsilon/8n$-tamed siblings of $\mathcal{P}$ and $\mathcal{Q}$ resp.}
				$\eta \gets \varepsilon/(\varepsilon+4)$\label{line:eta}\\
				$ m_{out} \gets \frac{\log(24)}{2} \left(\frac{1-\eta}{\eta}\right)^2
				$ \label{line:mout}\\
				$m_{in} \gets 32\log(48m_{out})$\label{line:min}\\
				$\mathtt{est}\gets 0$\\
				All{$i=1$ {\bfseries to} $m_{out}$}{
					
					$\sigma \gets \samp(\mathcal{Q'})$\\
					All{$j=1$ {\bfseries to} $m_{in}$}{
						$p_{j} \gets  \SubToAeval(\mathcal{P'}, \sigma, \eta)$\label{line:pi}\\
						$q_{j} \gets  \SubToAeval(\mathcal{Q'}, \sigma, \eta)$\label{line:qi}\\
					}
					$\hat{p} \gets \mathtt{Median}_{j} (p_{j})$\\
					$\hat{q} \gets \mathtt{Median}_{j} (q_{j})$\\
					\If{$ \hat{q}>\hat{p}$}{
						$\mathtt{est} \gets \mathtt{est} + 1 - \hat{p}/\hat{q}$
					}
				}
				\Return $\mathtt{est}/m_{out}$~\label{line:return}
			\end{algorithm2e}
		}
		
	}
	\setlength{\algomargin}{1em}
		\SetKwComment{Comment}{$\triangleright$ }{}
		\begin{algorithm2e}[ht] 
			\DontPrintSemicolon
			\caption{$\SubToAeval(\mathcal{D'},\sigma, \eta)$\label{alg:subtoaeval}}
			\Comment{$\mathcal{D'}$ is  $\varepsilon/8n$-tamed sibling of $\mathcal{D}$}
			$t \gets 0$\\
			$k \gets  4n\eta^{-2}(1+\eta^2) $\label{line:c}\\
			All{$i = 1$ {\bfseries to} $n$}{
				$x_i \gets 0$\\
				$f \gets 0$\\
				\While{$f < k$ \label{line:whilestart}}{   
					$\alpha \gets \condmar(\mathcal{D'},\sigma_{<i})$ \\
					$x_{i} \gets x_{i}+1$\\
					$t\gets t+ 1 $\\
					\lIf{ $t=   64n^3\eta^{-2}(1+\eta)^2\varepsilon^{-1} $\label{line:threshold}}{ \Return 0}
					\lIf{$\alpha= \sigma_i$}{$f \gets f+1$}	 \label{line:whileend}
				}
			}			
			$ d \gets  \prod_{i=1}^{n} k/x_{i}$\\
			\Return $d$
		\end{algorithm2e}
\subsection{High-Level Overview}
In Section~\ref{sec:outline-subvssub}, we introduce the main ideas of our algorithms, \SubVsSubMain and \SubVsSub. Then, in Section~\ref{sec:subtoaeval}, we explain the key concepts of the \SubToAeval subroutine.
\subsubsection{Outline of the \SubVsSubMain and   \SubVsSub routines} \label{sec:outline-subvssub}

The pseudocode of $\SubVsSubMain$ and $\SubVsSub$ is given in Alg.~\ref{alg:subvssubmain} and ~\ref{alg:subvssub} respectively. $\SubVsSubMain$ takes as input two distributions $\mathcal{P}$ and $\mathcal{Q}$ defined over the support $\bn$, along with the parameter $\varepsilon$ for tolerance and the parameter $\delta$ for confidence, and returns an $\varepsilon$-additive estimate of $d_{TV}(\mathcal{P},\mathcal{Q})$ with probability at least $1-\delta$. 

The \SubVsSub subroutine call returns an estimate $r_j$ of $d_{TV}(\mathcal{P},\mathcal{Q})$ such that $\Pr[r_j \in (d_{TV}(\mathcal{P},\mathcal{Q})\pm \varepsilon)] \geq 2/3$, and \SubVsSubMain makes  $48\log(1/\delta)$ calls to boost the overall probability to $1-\delta$, using the Chernoff bound on the median of the estimates.

$\SubVsSub$ takes as input the distributions $\mathcal{P}$ and $\mathcal{Q}$, and creates their $\varepsilon/8n$-tamed siblings $\mathcal{P'}$ and $\mathcal{Q'}$ that are $\varepsilon/8$ close to $\mathcal{P}$ and $\mathcal{Q}$ in TV distance, and have the property that all of their marginal probabilities are lower bounded by $\Omega(\varepsilon/8n)$. 
The bounded marginal property of $\mathcal{P}'$ and $\mathcal{Q'}$ is crucial for the polynomial query complexity of \SubVsSub. The construction of $\mathcal{P}'$ and $\mathcal{Q}'$, and the claimed guarantees, are discussed in Section~\ref{sec:taming}.
$\SubVsSub$ then computes the constants $\eta$, $m_{out}$, and $m_{in}$ (the counts of iterations of the outer and inner loop).

$\SubVsSub$ then draws $m_{out}$ samples $\sigma \sim \mathcal{Q'}$, and for each sample $\sigma$, calls \SubToAeval $m_{in}$ times to find the $(1\pm\eta)$ estimates of $\mathcal{Q'}(\sigma)$ and $\mathcal{P'}(\sigma)$. The \SubToAeval subroutine puts an upper limit on the number of \condmar oracle calls, and the limit is set high enough to ensure that the estimates, $\hat{p}$ and $\hat{q}$, are correct with the required confidence. \SubVsSub then computes the distance using these estimates as given in Lemma~\ref{lem:evalvseval}.

\subsubsection{Outline of the $\SubToAeval$ subroutine}\label{sec:subtoaeval}
The \SubToAeval subroutine takes as input an element $\sigma \in \bn$, a distribution $\mathcal{D}$ over $\bn$, and a parameter $\eta$.
\SubToAeval outputs an $\eta$-multiplicative estimate of $\mathcal{D}(\sigma)$. The probability $\mathcal{D}(\sigma)$ can be expressed as a product of marginals, $\mathcal{D}(\sigma) = \prod_{i=1}^{n} \mathcal{D}_{\sigma_{<i}}(\sigma_i)$, by applying the chain rule. Essentially, the subroutine  approximates each marginal $\mathcal{D}_{\sigma_{<i}}(\sigma_i)$ by $k/x_i$ for each $i \in [n]$, using the $\condmar$ oracle. The product $\prod_{i=1}^{n} k/x_{i}$ is then employed as the final estimate for $\mathcal{D}(\sigma)$. 

In this context, the variable $x_i$ represents the total count of $\condmar(\mathcal{D},\sigma_{<i})$ queries executed until $k$ occurrences of $\sigma_i$ are observed. Given that $\mathcal{D}_\rho(b) =  \Pr_{w \sim \condmar(\mathcal{D},{\rho})}[w = b]$ for any $\rho$ (as discussed in Section \ref{sec:notation}), the ratio $k/x_i$ is an intuitive choice as an estimator for $\mathcal{D}_{\sigma_{<i}}(\sigma_i)$. Moreover, to ensure the subroutine terminates, a total number of calls to the $\condmar$ oracle are monitored, and if they ever exceed the threshold $ 64n^3\eta^{-2}(1+\eta)^2\varepsilon^{-1}$, the subroutine terminates and returns $0$.

We now discuss our technical contribution - showing the correctness of \SubToAeval when the threshold is set to $O(n^3)$ (for this discussion, we will set aside the dependency on $\eta$). To estimate $\mathcal{D}(\sigma)$, it is essential to estimate each of the $n$ marginals, $\mathcal{D}_{\sigma_{<i}}(\sigma_i)$, to within an error margin of approximately $1+1/n$. This would require at least $n^2/\mathcal{D}_{\sigma_{<i}}(\sigma_i)$ queries for each marginal. Consequently, the total query complexity would sum up to $\sum_{i=1}^{n} n^2/\mathcal{D}_{\sigma_{<i}}(\sigma_i)$. This quantity is at least $\Omega(n^2)$, but it could potentially be unbounded as $\mathcal{D}_{\sigma_{<i}}(\sigma_i)$ can take arbitrarily small values. In the forthcoming section, we reduce this complexity to $O(n^3)$ through a more nuanced analysis.

\subsection{Theoretical Analysis}\label{sec:analysis}
In this section, we will prove our main Theorem~\ref{thm:scondvsscond}. The proof of Theorem~\ref{thm:scondvsscond} relies on Lemma~\ref{lem:subtoaeval},  which claims the correctness of the $\SubToAeval$ subroutine and upper bound its query complexity. We will prove the lemma later.

\begin{restatable}{lemma}{}\label{lem:subtoaeval}
	$\SubToAeval(\mathcal{D'}, \sigma, \eta)$ takes as input distribution $\mathcal{D'}$, $\sigma \in \bn$, $\eta \in (0,1/5)$ and returns $d$, then
	\begin{align*}
		\Pr[d  \in (1\pm \eta) \mathcal{D'}(\sigma) ] \geq 5/8
	\end{align*}
	\SubToAeval makes $ O(n^3/\eta^2)$ $\condmar$ queries to $\mathcal{D'}$. 
\end{restatable}

\scondvsscond*
\begin{proof}
	We will first show that  the algorithm $\SubVsSub(\mathcal{P},\mathcal{Q},\varepsilon)$  returns $\mathtt{est}$ such that 
	\begin{align*}
		\Pr[ \mathtt{est} \in (d_{TV}(\mathcal{P},\mathcal{Q})   \pm \varepsilon)  ] \geq 5/6
	\end{align*}
	
	Since \SubVsSubMain returns the median of the independent estimates provided by \SubVsSub, then applying the Chernoff bound, we have $\Pr[ \kappa \in (d_{TV}(\mathcal{P},\mathcal{Q}) \pm \varepsilon)] \geq 1-\delta$.

	We will now consider the events that could lead to an incorrect estimate. 
	Recalling that $\mathcal{P'}$ and $\mathcal{Q'}$ are $\varepsilon/8n$-tamed siblings of $\mathcal{P}$ and $\mathcal{Q}$ we define $\mathtt{Bad}_i^{\hat{p}}$ and $\mathtt{Bad}_i^{\hat{q}}$ to be the events that in the $i^{th}$ iteration  of \SubVsSub, $\hat{p}\not \in (1\pm \eta)\mathcal{P'}(\sigma)$, and  $\hat{q} \not \in (1\pm \eta)\mathcal{Q'}(\sigma)$, respectively.  We bound the probability of $\mathtt{Bad}_i^{\hat{p}}$ and $\mathtt{Bad}_i^{\hat{q}}$ in the following claim, whose proof is deferred to Appendix.
	
	\begin{restatable}{claim}{badip}\label{claim:badip}
		$\Pr[\mathtt{Bad}_{i}^{\hat{p}}]  \leq 1/24m_{out}$ and $\Pr[\mathtt{Bad}_{i}^{\hat{q}}]  \leq 1/24m_{out}$.
	\end{restatable}
	
	Now we define $\mathtt{Bad} = \bigcup_{i \in [m_{out}]}(\mathtt{Bad}_i^{\hat{p}} \cup \mathtt{Bad}_i^{\hat{q}})$, i.e.,  $\mathtt{Bad} $ captures the event that at least one of the estimates is incorrect. Then from Claim~\ref{claim:badip} and the union bound, 	  $\Pr[\mathtt{Bad}] = \Pr\left[\underset{i \in [m_{out}]}{\bigcup}(\mathtt{Bad}^{\hat{p}}_i \cup \mathtt{Bad}^{\hat{q}}_i)\right] 
	\leq \underset{i\in[m_{out}]}{\sum}(\Pr[\mathtt{Bad}^{\hat{p}}_i] + \Pr[\mathtt{Bad}^{\hat{q}}_i])
	\leq m_{out}\left( \frac{1}{24m_{out}} +\frac{1}{24m_{out}}\right)
	\leq  \frac{1}{12}$.
	
	Now, let's assume the event $\overline{\mathtt{Bad}}$. We have a set of $m_{out}$ samples from $\mathcal{Q}'$, and for each sample $\sigma$ we have $\hat{p}$ and $\hat{q}$ such that $\hat{p} \in (1\pm \eta) \mathcal{P'}(\sigma)$ and $\hat{q} \in (1\pm \eta) \mathcal{Q'}(\sigma)$.  This fulfills the condition of Lemma~\ref{lem:evalvseval}, and hence 
	substituting $|S| = m_{out}$ (Line~\ref{line:mout} of Alg.\ref{alg:subvssub}) we have, 
	\begin{align*}
		 &\Pr\left[  \mathtt{est} \not \in   \left(d_{TV}(\mathcal{P'},\mathcal{Q'}) \pm \frac{3\eta}{1-\eta}\right) \cap \overline{\mathtt{Bad}}  \right] 
		\\ &\leq    2\exp\left(-2m_{out}\left(\frac{\eta }{1-\eta}\right)^2\right) \leq  2\exp\left(- \log(24)\right) = \frac{1}{12}
	\end{align*}
	Substituting $\eta = \frac{\varepsilon}{\varepsilon+4}$ from  Alg.\ref{alg:subvssub},  we have, $	\Pr\left[ \mathtt{est}   \not \in   \left(d_{TV}(\mathcal{P'},\mathcal{Q'}) \pm \frac{3\varepsilon}{ 4}\right)\cap \overline{\mathtt{Bad}} \right] \leq \frac{1}{12}$. Then, 
	\begin{align*}	
		&\Pr\left[ \mathtt{est}   \not \in   \left(d_{TV}(\mathcal{P'},\mathcal{Q'}) \pm \frac{3\varepsilon}{ 4}\right)\right] \\
		&\leq 	\Pr\left[ \mathtt{est}   \not \in   \left(d_{TV}(\mathcal{P'},\mathcal{Q'}) \pm \frac{3\varepsilon}{ 4}\right) \cap \overline{\mathtt{Bad}} \right] + \Pr[\mathtt{Bad}]\\
		&\leq 1/12+1/12 = 1/6 
	\end{align*}
	
	Since $\mathcal{P}'$ and $\mathcal{Q'}$ are $\varepsilon/8n$-tamed siblings of $\mathcal{P}$ and $\mathcal{Q}$, from Lemma~\ref{lem:tame} we know that $d_{TV}(\mathcal{P}',\mathcal{P})\leq \varepsilon/8$ and $d_{TV}(\mathcal{Q'},\mathcal{Q})\leq \varepsilon/8$.	
	Then, from the triangle inequality, we have the  bounds on $ d_{TV}(\mathcal{P},\mathcal{Q}) $:
	\begin{align*}
		d_{TV}(\mathcal{P'},\mathcal{Q'}) & \in  d_{TV}(\mathcal{P},\mathcal{Q}) \pm \left(d_{TV}(\mathcal{P}',\mathcal{P}) +d_{TV}(\mathcal{Q'},\mathcal{Q})\right) \\
		&\in d_{TV}(\mathcal{P},\mathcal{Q}) \pm \varepsilon/4
	\end{align*}
	
	Combining the two, we get that $\Pr[ \mathtt{est} \not \in \left(d_{TV}(\mathcal{P},\mathcal{Q}) \pm \varepsilon \right) ] \leq 1/6$, and hence we have our claim. 
	
	Now, we will complete the proof by showing an upper bound on the query complexity.
	The total number of $\condmar$ queries made by $\SubToAeval(\mathcal{D'}, \sigma, \eta)$ in a single invocation is $ 64n^3\eta^{-2}(1+\eta)^2\varepsilon^{-1}= O(n^3\varepsilon^{-3})$. Then  $\SubVsSub(\mathcal{P},\mathcal{Q},\varepsilon)$ makes $m_{in}m_{out} = O(\varepsilon^{-2}\log(\varepsilon^{-1}))$ many calls to $\SubToAeval$. Finally,   \SubVsSubMain calls \SubVsSub $48\log(1/\delta)$ many times. Thus the total number of queries to the $\condmar$ oracle made by \SubVsSubMain is $O\left(n^3\log(1/\delta)\log(\varepsilon^{-1})/\varepsilon^5 \right)$.
\end{proof}

\begin{proof}[Proof of Lemma~\ref{lem:subtoaeval}]
	Consider the subroutine $\SubToAeval_1(\mathcal{D'},\sigma, \eta)$ (Alg.~\ref{alg:subtoaeval'}), that is the same as $\SubToAeval(\mathcal{D'}, \sigma, \eta)$ (Alg.~\ref{alg:subtoaeval}) except in one critical aspect:  the termination condition on Line~\ref{line:threshold} of \SubToAeval has been removed. This implies that while $\SubToAeval(\mathcal{D'},\sigma, \eta)$ terminates if the number of calls to the \condmar oracle exceeds the threshold $ 64n^3\eta^{-2}(1+\eta)^2\varepsilon^{-1}$,  $\SubToAeval_1(\mathcal{D'},\sigma, \eta)$ does not enforce this restriction, thereby allowing an unlimited number of calls to the \condmar oracle. 
	Note that we use variable names $d_1$ and $t_1$ in \SubToAeval' to distinguish them from $d$ of $t$ of \SubToAeval.
	This modification is critical for our analysis as it leads to the variable $x_i$ in $\SubToAeval_1(\mathcal{D'},\sigma, \eta)$ following the negative binomial distribution.

	{\setlength{\algomargin}{1em}
		\DontPrintSemicolon
		\begin{algorithm2e}[h]
			\caption{$\SubToAeval_1(\mathcal{D}',\sigma, \eta)$} \label{alg:subtoaeval'}
			$t_1 \gets 0$\\
			$k \gets  4n\eta^{-2}(1+\eta^2) \label{line:k}$\\
			All{$i = 1$ {\bfseries to} $n$}{
				$x_i \gets 0$\\
				$f \gets 0$\\
				\While{$f <k$}{   
					$\alpha \gets \condmar(\mathcal{D'},\sigma_{<i})$ \\
					$x_{i} \gets x_{i}+1$\\
					$t_1 \gets t_1 + 1$\\
					\lIf{$\alpha = \sigma_i$}{$f \gets f+1$}				
				}
			}			
			$ d_1 \gets  \prod_{i=1}^{n} k/x_{i}$ \label{line:outsideloopsubtoaeval'}\\
			\Return $d_1$
		\end{algorithm2e}
	}

	\begin{remark}
		Henceforth we will use $t_1$ and $x_i$ to denote the final values of $t_1$ and $x_i$, as on Line~\ref{line:outsideloopsubtoaeval'}.
	\end{remark}
	
	We will now show that the  $\SubToAeval_1(\mathcal{D'},\sigma, \eta)$ correctly estimates $\mathcal{D'}(\sigma)$ with high probability  (Lemma~\ref{lem:d1bound}) and then we show that it makes fewer than $  64n^3\eta^{-2}(1+\eta)^2\varepsilon^{-1}$  calls to \condmar oracle with high probability (Lemma~\ref{lem:t1bound}).  These results will help us establish analogous results for the subroutine $\SubToAeval(\mathcal{D},\sigma, \eta )$ and in validating our Lemma~\ref{lem:subtoaeval}.
	
	\begin{observation}\label{obs:mainobservation}
		Comparing \SubToAeval and $\SubToAeval_1$, we observe that \SubToAeval returns an incorrect estimate $d$ in two cases. Either \SubToAeval returns incorrect $d_1$, or else $\SubToAeval_1$ makes more than $64n^3\eta^{-2}(1+\eta)^2 \varepsilon^{-1}$ queries. Stated formally,
		\begin{align*}
			&\Pr[d \not \in (1\pm \eta)\mathcal{D'}(\sigma)] \\
			&\leq  \Pr\left[d_1 \not \in (1\pm \eta)\mathcal{D'}(\sigma)\right] + \Pr\left[t_1 \geq  64n^3\eta^{-2}(1+\eta)^2 \varepsilon^{-1}\right]
		\end{align*}
	\end{observation}
	
	Our proof will use the  following prop. and lemmas.
	\begin{restatable}{proposition}{negbin}
		\label{prop:negative-binomial}
		For $i \in [n]$,  the value of $x_i$ (in Alg.~\ref{alg:subtoaeval'}) is distributed as   $\nb(k, \mathcal{D}_{\sigma_{<i}}(\sigma_i))$
	\end{restatable}
We  prove  the above proposition in Appendix~\ref{sec:appendix:d1t1bound}.
	\begin{restatable}{lemma}{dbound}\label{lem:d1bound}
		$\Pr[ d_1 \in (1\pm \eta) \mathcal{D'}(\sigma)] \geq 2/3$.
	\end{restatable} 
	\begin{proof}
		We use a variance reduction technique introduced by \citet{DF91}.  $x_i$ on Line~\ref{line:outsideloopsubtoaeval'} is distributed according to $\nb(k, \mathcal{D'}_{\sigma_{<i}}(\sigma_i))$, so we have $\expect[x_i]= k/\mathcal{D'}_{\sigma_{<i}}(\sigma_i)$, and hence,  $k/\expect[x_i]= \mathcal{D'}_{\sigma_{<i}}(\sigma_i)$.  Now since $d_1=\prod_{j=1}^{n}k/x_i$, we have  $\expect[1/d_1] = \expect[\prod_{i=1}^{n} x_i/k]=\prod_{i=1}^{n}1/\mathcal{D'}_{\sigma_{<i}}(\sigma_i)$.
		\begin{align*}
			\frac{\variance[1/d_1]}{\expect[1/d_1]^2} =  \frac{\expect[1/d_1^2]}{\expect[1/d_1]^2} -1 &= \prod_{i=1}^{n} \frac{\expect[(x_i/k)^2]}{\expect[x_i/k]^2} -1 \\
		&	= \prod_{j=1}^{n} \left(1+\frac{\variance[x_i/k]}{\expect[x_i/k]^2}\right) -1 
		\end{align*}
		Using the fact that $x_i$ is negative binomial, we substitute $\variance[x_i/k]$ and $\expect[x_i/k]^2$, 
		\begin{align*}
		&	\frac{\variance[1/d_1]}{\expect[1/d_1]^2} = \prod_{j=1}^{n} \left(1+\frac{(1-\mathcal{D'}_{\sigma_{<i}})/k\mathcal{D'}_{\sigma_{<i}}^2}{(1/\mathcal{D'}_{\sigma_{<i}})^2}\right) -1 \\
		&	= \prod_{j=1}^{n} \left(1+\frac{1-\mathcal{D'}_{\sigma_{<i}}(\sigma_i)}{k}\right) -1 \leq \prod_{j=1}^{n} \left(1+\frac{1}{k}\right) -1
		\end{align*}
		Substituting the value of $k$ from the algorithm, we have 
		\begin{align}
			\frac{\variance[1/d_1]}{\expect[1/d_1]^2}  & \leq  \left(1 + \frac{\eta^2}{4n(1+\eta)^2} \right)^n - 1  \nonumber\\
			&\leq \exp\left(\frac{\eta^2}{4} \right) - 1\leq \frac{\eta^2}{3(1+\eta)^2} \label{line:dispersion}
		\end{align}
		The last inequality comes from the fact that for $r\in(0,1),s>1$,  $\exp\left(\frac{r}{s+1}\right) \leq 1+\frac{r}{s}$.  Recall that from the chain rule we have $\mathcal{D'}(\sigma) = \prod_{j=1}^n\mathcal{D'}_{\sigma_{<i}}(\sigma_i)$, then  $\expect[1/d_1] = 1/\mathcal{D'}(\sigma)$. 
		\begin{align}
			&\Pr[d_1  \in (1\pm \eta) \mathcal{D'}(\sigma)] 
			\nonumber =   \Pr\left[\frac{1}{d_1} \in \left[\frac{1}{1+ \eta}, \frac{1}{1-\eta} \right]\frac{1}{\mathcal{D'}(\sigma)}\right]  \\ 
			\nonumber	 & =   \Pr\left[  \frac{1}{d_1} - \expect\left[\frac{1}{d_1}\right] \in \left[-\frac{\eta}{1+ \eta}, \frac{\eta}{1-\eta}\right] \expect\left[\frac{1}{d_1}\right]  \right] \\
			\nonumber	 & \geq   \Pr\left[  \left|\expect\left[\frac{1}{d_1}\right]-\frac{1}{d_1}\right| \leq  \frac{\eta}{1+\eta}\expect\left[\frac{1}{d_1}\right]  \right] \\
			&\geq 1- \frac{(1+\eta)^2}{\eta^2}\frac{\variance\left[\frac{1}{d_1}\right]}{\expect \nonumber\left[\frac{1}{d_1}\right]^2} \geq 1- \frac{1}{3} = \frac{2}{3} 
		\end{align}
		We use the Chebyshev bound to get the second to last inequality and then substitute (\ref{line:dispersion}).
	\end{proof}
	\begin{table*}[t]
\centering
\caption{The sample complexity and runtime performance of \SubVsSubMain on real-world instances.}
\begin{tabular}{lccccc}
\toprule
Benchmark & Dimensions & \multicolumn{2}{c}{\sts} & \multicolumn{2}{c}{\cmsgen} \\ \cmidrule(lr){3-4} \cmidrule(lr){5-6}
& & \# of samples & time (in s) & \# of samples & time (in s) \\
\midrule
 s1196a\_3\_2 & 33 & 1.8e+09 & 4.1e+05 & 1.9e+09 & 5.3e+05\\
 53.sk\_4\_32 & 33 & 1.7e+09 & 2.5e+05 & 1.9e+09 & 1.6e+06\\
27.sk\_3\_32 & 33 & 1.7e+09 & 1.9e+05 & 1.9e+09 & 1.0e+06\\
s1196a\_7\_4 & 33 & 1.8e+09 & 4.6e+05 & 1.9e+09 & 5.5e+05\\
s420\_15\_7 & 35 & 2.1e+09 & 4.2e+05 & 2.3e+09 & 4.0e+05\\
111.sk\_2\_36 & 37 & 2.2e+09 & 3.5e+05 & 8.3e+08 & 6.6e+05\\
\bottomrule
\end{tabular}
\label{table}
\end{table*}

	Note that in every iteration,  $t_1$ gets incremented by the value of $x_i$. In the following lemma, we claim that $t_1$, the number of queries made by $\SubToAeval_1$, exceeds the threshold on Line~\ref{line:threshold} of \SubToAeval with low probability. We defer the proof to the Appendix~\ref{sec:appendix:d1t1bound}.
	\begin{restatable}{lemma}{tbound}\label{lem:t1bound}
		$\Pr[t_1 \geq  64n^3\eta^{-2}(1+\eta)^2\varepsilon^{-1}] \leq 1/24$
	\end{restatable}
	
	Putting together lemmas \ref{lem:d1bound} and \ref{lem:t1bound}  along with the observation \ref{obs:mainobservation} , we complete the proof:
	\begin{align*}
		&\Pr[d \not \in (1\pm \eta)\mathcal{D'}(\sigma)] \\
		&\leq  \Pr\left[d_1 \not \in (1\pm \eta)\mathcal{D'}(\sigma)\right] + \Pr\left[t_1 \geq  64n^3\eta^{-2}(1+\eta)^2\varepsilon^{-1}\right] \\
		&\leq \frac{1}{3} + \frac{1}{24} = \frac{3}{8}
	\end{align*}
\end{proof}

\subsection{The Discrete Hypergrid $\Sigma^n$} \label{sec:hypergrids}
This section extends our results beyond the hypercube $\{0,1\}^n$ to the hypergrid $\Sigma^n$, where $\Sigma$ is any discrete set. This line of investigation is motivated by the fact that in modern ML, distributions models are frequently described over hypergrids. For instance, language models are defined to be distributions over $\Sigma^n$ where $\Sigma$ is the set of tokens, and $n$ the length of the generated string. Furthermore, the \subcond oracle is particularly suitable for use in ML applications as it models autoregressive generation.

 The \subcond oracle for $\mathcal{D}$ supported on $\Sigma^n$, takes a query string $\rho \in \{\Sigma \cup*\}^n$ and draws samples from the set of strings that match all the non-$*$ characters of $\rho$. As noted in~\citep{CM23}, algorithms for $\{0,1\}^n$ do not immediately translate into algorithms for $\Sigma^n$, because the \subcond oracle does not work with the natural reduction of replacing elements $c \in \Sigma$ with their binary encoding. Nevertheless,  \SubVsSub can be extended to distributions over $\Sigma^n$, incurring a linear dependence on $|\Sigma|$.

We will now restate our result adapted to the new setting:
\begin{restatable}{theorem}{hypergrid-scondvsscond}\label{thm:hypergrid-scondvsscond}
	Given two distributions $\mathcal{P}$ and $\mathcal{Q}$ over $\Sigma^n$, along with parameters $\varepsilon \in( 0,1)$, $\delta \in (0,1/2)$, the algorithm $\SubVsSubMain(\mathcal{P},\mathcal{Q},\varepsilon,\delta)$, and with probability at least $1-\delta$ returns $\kappa$, s.t. $$ \Pr[  \kappa \in (d_{TV}(\mathcal{P},\mathcal{Q}) \pm \varepsilon)  ] \geq 1-\delta$$
	$\SubVsSubMain(\mathcal{P},\mathcal{Q},\varepsilon,\delta)$ makes $\Tilde{O}\left(n^3|\Sigma|\log(1/\delta)/\varepsilon^5 \right)$ $\subcond$ queries.
\end{restatable}

The only change required in $\SubVsSubMain$ to make it work for distributions over the $\Sigma^n$  is in the construction of the tamed siblings $\mathcal{P'}$, and $\mathcal{Q'}$.  We update the taming parameter from $ \varepsilon/8n$ to $\varepsilon/8n |\Sigma|$. Since the query complexity is proportional to $1/\theta$, we observe a linear dependence on $|\Sigma|$.

\section{EXPERIMENTS}\label{sec:experiments}

We implemented \SubVsSubMain in Python. We focus on  distributions generated by state-of-the-art combinatorial samplers \sts\citep{EGS12} and \cmsgen \citep{GSCM21}.
Our assessment included two datasets: (1) \texttt{scalable} comprising random Boolean functions over $n$ variables, with $n$ ranging from 30 to 70, and (2) \texttt{real-world}, containing instances from the ISCAS89 dataset, a standard in combinatorial testing and sampling evaluations~\citep{M20}.
To determine the ground truth TV distance for the above instances, we implement a learning-based distance estimator~\cite{noteC20}. 

For our experiments, we set the tolerance  $\varepsilon = 0.3$ and confidence $\delta = 0.4$ as the default throughout the evaluation. These parameters indicate that the estimate returned by \SubVsSubMain is expected to be within $\pm0.3$ of the ground truth, with a probability of at least $ 0.6$.

The experiments were conducted on a cluster with AMD EPYC 7713 CPU cores. We use 32 cores with 4GB of memory for each benchmark and a 24-hour timeout per instance.

Our aim was to answer the question: To what extent does \SubVsSubMain scale, i.e., how many dimensions can the estimator handle while providing guarantees?

We found that \SubVsSubMain scales to $n=70$ dimensional problems, a regime where the baseline sample-based estimators would require $10^7\times$ more samples. The estimates are empirically confirmed to be of high quality when compared against the ground truth, falling within the allowed tolerance bound in all cases where we could determine the ground truth.

Table~\ref{table} details the performance of \SubVsSubMain on 6  \texttt{real-world} benchmarks. The algorithm successfully finished on all benchmarks with dimensionality up to $n=37$. The table specifies the benchmark name, dimensionality, sample count, and processing time for both \sts and \cmsgen.

The sample complexity of \SubVsSubMain relative to a baseline sample-based estimator is illustrated in Figure~\ref{fig:valiant}(in Section~\ref{sec:intro}). For this, we use \texttt{scalable} benchmarks. Remarkably, for the largest instance handled ($n=70$ dimensions), \SubVsSubMain outperformed the baseline by a factor greater than $10^7$.

\section{CONCLUSION}\label{sec:conclusion}
This paper focused on the distance estimation problem in the \subcond model. We sought to alleviate the significant weakness of the prior state of the art: the estimators required an exponentially large number of queries. Our primary contribution, \SubVsSubMain, enables distance estimation in $\mathcal{O}(n^3/\varepsilon^5)$ queries. We also implemented \SubVsSubMain and tested it on distributions generated by combinatorial samplers, showing the scalability of our approach. The problem of closing the gap between the $\mathcal{O}(n^3/\varepsilon^5)$ upper bound and the $\Omega(n/\log(n))$ lower bound, remains open in all \cond models.

\section{ACKNOWLEDGEMENTS}\label{sec:ack}

 This research is part of the programme DesCartes and is supported by the National Re-
search Foundation, Prime Minister’s Office, Singapore, under its Campus for Research Excellence and Technological Enterprise (CREATE) programme. The computational
works of this article were performed on the resources of the National Supercomputing Centre, Singapore(\url{www.nscc.sg}).

The authors decided to forgo the old convention of alphabetical ordering of authors in favor of a randomized ordering, denoted by \textcircled{r}.

\bibliographystyle{plainnat}
\bibliography{main}

\newpage
\appendix
\onecolumn
\section{Appendix}\label{sec:appendix}
\subsection{Useful Inequalities}\label{sec:appendix:useful}
\begin{lemma}[Chernoff]
For any $\gamma,\delta \in(0, 0.5]$, let $n \geq \log(2/\delta)/2\gamma^2$, and let $x_1, x_2, \ldots, x_n$ be i.i.d. variables taking value in $(0, 1]$, with mean $\expect[x]$, then 
		\begin{align*}
\Pr\left[\left|\frac{1}{n}\sum_{i=1}^{n}x_i - \expect[x] \right| <\gamma \right] \geq 1-\delta
		\end{align*}
\end{lemma}
\begin{lemma}[Chebyshev]
	Let $X$ be a random variable with $\expect[X^2] < \infty$. For any $t >  0$, we have $\Pr[|\expect[X]-X| \leq t]\leq  \variance[X]/t^2$
\end{lemma}

\subsection{Lower Bound}\label{sec:lboundforsvs}
To complement the upper bound shown in the main paper, we show that the best-known lower bound 
for the problem in Theorem~\ref{thm:scondvsscond}, is $\Omega(n/\log(n))$. This bound is from~\citet{CDKS20}.

 \begin{theorem}[Theorem 11 in \citep{CDKS20}]
 	 An absolute constant $\varepsilon_0 < 1$ exists, such as the following holds.
 	Any algorithm that, given a parameter $\varepsilon \in (0,\varepsilon_0] $, and sample access to product distributions $\mathcal{P},\mathcal{Q}$ over $\{0,1\}^n$, distinguishes between $d_{TV}(\mathcal{P}, \mathcal{Q})  < \varepsilon$ and $d_{TV}(\mathcal{P}, \mathcal{Q}) > 2\varepsilon$, with probability at least 2/3, requires $\Omega(n/\log(n))$ samples.  Moreover, the lower bound still holds in the case where $\mathcal{Q}$ is known, and provided as an explicit parameter.
 \end{theorem}
  The lower bound is shown for the case where the tester has access to samples from a product distribution $\mathcal{P}$ and $\mathcal{Q}$(over $\{0,1\}^n$). As observed by~\citet{BC18}, {\subcond} access is no stronger than {\samp} when it comes to product distributions. Thus we have the following lower bound: 
 \begin{cor}
 	Let $\mathcal{S}(\varepsilon_1,\varepsilon_2,\mathcal{P},\mathcal{Q})$ be any algorithm that has \subcond access to distribution $\mathcal{P}$,  and explicit knowledge of $\mathcal{Q}$ (defined over $\{0,1\}^n$), and  distinguishes between $d_{TV}(\mathcal{P},\mathcal{Q}) \leq \varepsilon_1$  and $d_{TV}(\mathcal{P}, \mathcal{Q}) > \varepsilon_2$ with probability $>2/3$. Then,  $\mathcal{S}$ makes $\Omega(n/\log(n))$ \subcond queries.
 \end{cor}

\newpage

\subsection{Proof of Lemma~\ref{lem:evalvseval}}\label{sec:appendix:evalvseval}
\evalvseval*
\begin{proof}
	Recall that $p_\sigma \in (1\pm\eta) \mathcal{P}(\sigma)$ and $q_\sigma \in (1\pm\eta)\mathcal{Q}(\sigma)$  then, using the definition of $d_{TV}(\mathcal{P},\mathcal{Q})$,
	\begin{align}
		d_{TV}(\mathcal{P}, \mathcal{Q}) &= \underset{\sigma \in \bn}{\sum}\mathds{1}_{\mathcal{Q}(\sigma) > \mathcal{P}(\sigma)} \left(1-\frac{\mathcal{P}(\sigma)}{\mathcal{Q}(\sigma)}\right)\mathcal{Q}(\sigma) \nonumber\\
		&= \underset{\sigma \in \bn}{\sum}\mathds{1}_{q_\sigma > p_\sigma} \left(1-\frac{p_\sigma}{q_\sigma}\right)\mathcal{Q}(\sigma) \label{line:dtvexpand} \\
		&+\underbrace{\underset{\sigma \in \bn}{\sum} \left( \mathds{1}_{\mathcal{Q}(\sigma) > \mathcal{P}(\sigma)} \left(1-\frac{\mathcal{P}(\sigma)}{\mathcal{Q}(\sigma)}\right)\mathcal{Q}(\sigma) - \mathds{1}_{q_\sigma > p_\sigma} \left(1-\frac{p_\sigma}{q_\sigma}\right)\mathcal{Q}(\sigma)\right)}_A \nonumber
	\end{align}
	
	The first summand of (\ref{line:dtvexpand}) can be written as $\expect_{\sigma \sim \mathcal{Q}}\left[\mathds{1}_{q_\sigma > p_\sigma}\left(1-\frac{p_\sigma}{q_\sigma}\right)\right]$.

	To bound $|A|$, we will split the domain into three sets, $B_1 = \{x: \mathds{1}_{\mathcal{Q}(\sigma) > \mathcal{P}(\sigma)} = \mathds{1}_{q_\sigma > p_\sigma} \} $, $B_2 = \{x: \mathds{1}_{\mathcal{Q}(\sigma) > \mathcal{P}(\sigma)} > \mathds{1}_{q_\sigma > p_\sigma} \} $ and $B_3 = \{x: \mathds{1}_{\mathcal{Q}(\sigma) > \mathcal{P}(\sigma)} < \mathds{1}_{q_\sigma > p_\sigma} \} $. 
	\begin{align*}
		|A| &=  \left| \underset{\sigma \in \bn}{\sum}\left( \mathds{1}_{\mathcal{Q}(\sigma) > \mathcal{P}(\sigma)} \left(1-\frac{\mathcal{P}(\sigma)}{\mathcal{Q}(\sigma)}\right)\mathcal{Q}(\sigma) - \mathds{1}_{q_\sigma > p_\sigma} \left(1-\frac{p_\sigma}{q_\sigma}\right)\mathcal{Q}(\sigma)\right) \right| \\
		&\leq \underset{\sigma \in \bn}{\sum} \left| \left( \mathds{1}_{\mathcal{Q}(\sigma) > \mathcal{P}(\sigma)} \left(1-\frac{\mathcal{P}(\sigma)}{\mathcal{Q}(\sigma)}\right)\mathcal{Q}(\sigma) - \mathds{1}_{q_\sigma > p_\sigma} \left(1-\frac{p_\sigma}{q_\sigma}\right)\mathcal{Q}(\sigma)\right) \right|\\
		&	= \underset{\sigma \in B_1}{\sum} \mathds{1}_{\mathcal{Q}(\sigma) > \mathcal{P}(\sigma)} \left|\frac{\mathcal{P}(\sigma)}{\mathcal{Q}(\sigma)} - \frac{p_\sigma}{q_\sigma}\right|\mathcal{Q}(\sigma)	+ \underset{\sigma \in B_2}{\sum} \mathds{1}_{\mathcal{Q}(\sigma) > \mathcal{P}(\sigma)} \left(1-\frac{\mathcal{P}(\sigma)}{\mathcal{Q}(\sigma)}\right)\mathcal{Q}(\sigma) \\
		&+ \underset{\sigma \in B_3}{\sum} \mathds{1}_{q_\sigma > p_\sigma} \left(1-\frac{p_\sigma}{q_\sigma}\right)\mathcal{Q}(\sigma)
	\end{align*}
	
	For $\sigma \in B_1$, $\left|\frac{\mathcal{P}(\sigma)}{\mathcal{Q}(\sigma)} - \frac{p_\sigma}{q_\sigma}\right| \leq \frac{2\eta }{1-\eta}\frac{\mathcal{P}(\sigma)}{\mathcal{Q}(\sigma)}\leq \frac{2\eta }{1-\eta}$. 
	For $\sigma \in B_2$, $1 - \frac{\mathcal{P}(\sigma)}{\mathcal{Q}(\sigma)}\leq 1-\frac{1-\eta}{1+\eta} = \frac{2\eta}{1+\eta}$, and for $\sigma \in B_3$, $1 - \frac{\mathcal{P}(\sigma)}{\mathcal{Q}(\sigma)}\leq 1-\frac{1-\eta}{1+\eta} = \frac{2\eta}{1+\eta}$. Thus, $|A| \leq \sum_{\sigma \in B_1} \frac{2\eta}{1-\eta} \mathcal{Q}(\sigma)
	+ \sum_{\sigma \in B_2} \frac{2\eta}{1+\eta} \mathcal{Q}(\sigma) + \sum_{\sigma \in B_3} \frac{2\eta}{1+\eta} \mathcal{Q}(\sigma) \leq \frac{2\eta}{1+\eta}$. Plugging the bounds on $|A|$ back into (\ref{line:dtvexpand}), we get 
 \begin{align}
      \left| d_{TV}(\mathcal{P}, \mathcal{Q}) - \expect \left[\mathds{1}_{q_\sigma > p_\sigma}\left(1-\frac{p_\sigma}{q_\sigma}\right)\right] \right| \leq \frac{2\eta }{1-\eta} \label{line:dtvexpectation}
 \end{align}
 
And hence, $\expect \left[\mathds{1}_{q_\sigma > p_\sigma}\left(1-\frac{p_\sigma}{q_\sigma}\right)\right]-\frac{2\eta }{1-\eta}\leq d_{TV}(\mathcal{P},\mathcal{Q}) \leq \expect \left[\mathds{1}_{q_\sigma > p_\sigma}\left(1-\frac{p_\sigma}{q_\sigma}\right)\right]+\frac{2\eta }{1-\eta}$. 	
	The distance estimation algorithm draws $|S|$ samples to estimate $\expect \left[\mathds{1}_{q_\sigma > p_\sigma}\left(1-\frac{p_\sigma}{q_\sigma}\right)\right]$. We will use $\mathtt{est}$ to denote the empirical estimate of $\expect \left[\mathds{1}_{q_\sigma > p_\sigma}\left(1-\frac{p_\sigma}{q_\sigma}\right)\right]$. Since each  sample $\sigma$ is drawn independently, and  $\mathds{1}_{q_\sigma > p_\sigma}\left(1-\frac{p_\sigma}{q_\sigma}\right)$ is bounded in $[0,1]$, we can use the Hoeffding bound  as follows,
	\begin{align}
		\Pr\left[\left|\mathtt{est} - \expect \left[\mathds{1}_{q_\sigma > p_\sigma}\left(1-\frac{p_\sigma}{q_\sigma}\right)\right]\right|\geq \frac{\eta }{1-\eta} \right]
		&\leq 1-2\exp\left(-2|S|\left(\frac{\eta }{1-\eta}\right)^2\right) \label{line:hoeffding}
	\end{align}
	
	Plugging (\ref{line:dtvexpectation})  into   (\ref{line:hoeffding}), we complete the proof:
	\begin{align*}
	\Pr\left[\left|\mathtt{est} -d_{TV}(\mathcal{P},\mathcal{Q}) \right|\geq \frac{3\eta }{1-\eta} \right] = 
		\Pr\left[d_{TV}(\mathcal{P},\mathcal{Q}) \not \in \left( \mathtt{est} \pm \frac{3\eta }{1-\eta}\right) \right] 
&\leq 2\exp\left(-2|S|\left(\frac{\eta }{1-\eta}\right)^2\right)
	\end{align*}

\end{proof}
\newpage

\section{Proof of Lemma~\ref{lem:tame}}\label{sec:appendix:tame}
\tame*
\begin{proof}
Our proof  adapts the $\theta$-balancing trick, devised for product distributions in~\citet[Thm. 6]{CDKS20}. 
To simulate the $\condmar(\mathcal{D'},\sigma_{<\ell})$ query using $\subcond$ access to $\mathcal{D}$, we use the following process:
	$all i \geq \ell$, given the substring $\sigma_{< i}$, set  $\sigma_i = 0$ with probability $(1-2\theta)\mathcal{D}_{\sigma_{< i}}(0) + \theta$ and $\sigma_i = 1$ with probability $(1-2\theta)\mathcal{D}_{\sigma_{< i}}(1) + \theta$. To implement the above, with probability $1-2\theta$, draw  $\rho \sim \subcond(\mathcal{D}, \sigma_{< i})$ 	and return $\rho_i$, else with probability $2\theta$ draw a sample uniformly from $\ba$. 	
	
	 Observe that $all \ell \in [n]$, $c \in \{0,1\}$, and $\rho \in \{0,1\}^{\ell-1}$, we have  $ \mathcal{D'}_{\rho}(c) = (1-2\theta)\mathcal{D}_{\rho}(c) + \theta$. Since  $ \theta \leq \mathcal{D'}_{\rho}(c) \leq 1- \theta$,  we see that $\mathcal{D'}$ is indeed $\theta$-tamed.  To simulate $\samp(\mathcal{D'})$, we use the chain rule.
  
 Now we will show that $\mathcal{D'}$ is close to $\mathcal{D}$.
\begin{claim}
For distribution $\mathcal{D}$ and 	its $\theta$-tamed sibling $\mathcal{D'}$, we have $d_{TV}(\mathcal{D},\mathcal{D}') \leq \theta n$
\end{claim}
\begin{proof}
 Recall the definition of subcube $S_\rho = \{w \in \{0,1\}^n: w_{\le |\rho|} = \rho\}$. 
 For any set $S \subseteq \{0,1\}^n$, $\mathcal{D}(S)$ is the total probability of $S$ in $\mathcal{D}$. 
  For any distribution $\mathcal{D},$ string  $\rho$ (with $ 1 \le |\rho| \le n$) and $\omega \in \{0,1\}^{n-|\rho|}$, the distribution  $\mathcal{D}^{\rho}$ denotes the marginal distribution of $\subcond(\mathcal{D},{\rho})$ in the remaining dimensions,  i.e. for any $\omega \in \{0,1\}^{n-|\rho|}$,  $\mathcal{D}^{\rho}(\omega) =  \Pr_{w \sim \subcond(\mathcal{D},{\rho})}[w = \rho \omega]$.

Consider the induction hypothesis that $d_{TV}(\mathcal{D},\mathcal{D'}) \leq \theta i $ if $\mathcal{D}$ is supported on $\{0,1\}^i$. 
 To verify the hypothesis for $i=1$, wlog assume that $\mathcal{D}(0) \leq \mathcal{D}(1)$, then $d_{TV}(\mathcal{D},\mathcal{D'}) =  \mathcal{D}(1) -  \mathcal{D}'(1) =2\theta \mathcal{D}(1)-  \theta \leq \theta $.
	Assume the hypothesis holds for all $i \in [n-1]$. Now, we show the hypothesis is true for $i = n$. 
	
	Consider a distribution $\mathcal{D}$ over $\{0,1\}^n$ and its $\theta$-tamed sibling $\mathcal{D}'$, then:
	
	\begin{align*}
		&d_{TV}(\mathcal{D}, \mathcal{D}')  = \frac12 \sum_{\sigma \in \{0,1\}^{n}}|\mathcal{D}(\sigma) - \mathcal{D'}(\sigma)|
		=\frac12 \sum_{\rho \in \{0,1\}}\sum_{\omega \in \{0,1\}^{n-1}}|\mathcal{D}(\rho\omega) - \mathcal{D}'(\rho\omega)|
		\\
		& = \frac12\sum_{\rho \in \{0,1\}}\sum_{\omega \in \{0,1\}^{n-1}} |\mathcal{D}(S_{\rho})\mathcal{D}^{\rho}(\omega) - \mathcal{D}'(S_{\rho})\mathcal{D}'^{\rho}(\omega)|  \\
		&= \frac12 \sum_{\rho \in \{0,1\}}\sum_{\omega \in \{0,1\}^{n-1}} |\mathcal{D}(S_{\rho})\mathcal{D}^{\rho}(\omega) - \mathcal{D}(S_{\rho})\mathcal{D}'^{\rho}(\omega) +\mathcal{D}(S_{\rho})\mathcal{D}'^{\rho}(\omega) - \mathcal{D}'(S_{\rho})\mathcal{D}'^{\rho}(\omega)|\\
  	&\leq \frac12 \sum_{\rho \in \{0,1\}}\sum_{\omega \in \{0,1\}^{n-1}} |\mathcal{D}(S_{\rho})\mathcal{D}^{\rho}(\omega) - \mathcal{D}(S_{\rho})\mathcal{D}'^{\rho}(\omega)| +|\mathcal{D}(S_{\rho})\mathcal{D}'^{\rho}(\omega) - \mathcal{D}'(S_{\rho})\mathcal{D}'^{\rho}(\omega)|\\
  &=  \frac12\sum_{\rho \in \{0,1\}}\sum_{\omega \in \{0,1\}^{n-1}} \mathcal{D}(S_{\rho})|\mathcal{D}^{\rho}(\omega) - \mathcal{D}'^{\rho}(\omega)| +\mathcal{D}'_{\rho}(\omega)|\mathcal{D}'(S_{\rho}) - \mathcal{D}(S_{\rho})|\\
		&= \frac12\sum_{\rho \in \{0,1\}} \left(\mathcal{D}(S_{\rho})2d_{TV}(\mathcal{D}^{\rho},\mathcal{D}'^{\rho})\right)+\frac12 \sum_{\rho \in \{0,1\}}|\mathcal{D}'(S_{\rho}) - \mathcal{D}(S_{\rho})|\\
		&\leq  \sum_{\rho \in \{0,1\}} \left(\mathcal{D}(S_{\rho})\theta(n-1)\right) +\theta = \theta n
	\end{align*}
 We use $|a+b|\leq |a|+|b|$ in the first inequality. In the second, we use the induction hypothesis to bound the first summand, and for the second, we observe that for $c \in \{0,1\}$, $|\mathcal{D}'(c)-\mathcal{D}(c)| \leq \theta$.
\end{proof}\end{proof}

\newpage

\section{Proof of Claim~\ref{claim:badip}, Proposition~\ref{prop:negative-binomial} and Lemma~\ref{lem:t1bound}} \label{sec:appendix:d1t1bound}
\badip*
\begin{proof}
		For a fixed iteration $j$, applying  Lemma~\ref{lem:subtoaeval} we have $\Pr[p_j \in \left(1\pm \eta \right)\mathcal{P'}(\sigma)] \geq 5/8$. Since $\hat{p}$ is the median of independent observations $p_j \in [0,1]$, over  $j \in [m_{in}]$, we can use the Chernoff bound to derive the claimed bound,	$\Pr[\mathtt{Bad}_{i}^{\hat{p}}]  \leq 1/24m_{out}$. The proof for the claim $\Pr[\mathtt{Bad}_{i}^{\hat{q}}]  \leq 1/24m_{out}$ proceedes identically. 
\end{proof}

\negbin*
\begin{proof}
	Fix any $i \in [n]$. In  Alg.~\ref{alg:subtoaeval'}, the r.v $\alpha$ takes the value $\sigma_i$ with probability $ \mathcal{D}_{\sigma_{<i}}(\sigma_i)$.  Note that while the value of $x_i$ increments by one in every iteration of the loop (lines~\ref{line:whilestart}-\ref{line:whileend}), while the value of $f$ increases by one only when  $\alpha = \sigma_i$. Since the loop runs until the value of $f$ is $k$, the distribution of $x_i$ is $\nb(k, \mathcal{D}_{\sigma_{<i}}(\sigma_i))$.
\end{proof}

\tbound*
\begin{proof}
	The number of \condmar calls made by $\SubToAeval_1$ in the $i^{th}$ iteration is captured by $x_i$.
	Recall from Prop~\ref{prop:negative-binomial} that $x_i$ is drawn from $\nb(k,\mathcal{D'}_{\sigma_{<i}}(\sigma_i))$, and therefore we have,
	\begin{align*}
		\expect[x_i] = k/\mathcal{D'}_{\sigma_{<i}}(\sigma_i) = 4n\eta^{-2}(1+\eta)^2/\mathcal{D'}_{\sigma_{<i}}(\sigma_i) \qquad \text{(Using $k$ from Line~\ref{line:k} of $\SubToAeval_1$)} 
	\end{align*} 
	
	From the fact that the distribution is $\varepsilon/8n$-tamed, we know that $\mathcal{D'}_{\sigma_{<i}}(\sigma_i)\geq  \varepsilon/8n$. Hence we have  $\expect[x_i] \leq  32n^2\eta^{-2}(1+\eta)^2\varepsilon^{-1}$.  Since $t_1 = \sum_{i \in [n]}x_i$, we have that
	$\expect[t_1] = \expect[\sum_{i \in [n]} x_i] =  n\expect[x_i] \leq 32n^3\eta^{-2}(1+\eta)^2\varepsilon^{-1}$. Thus, 
	\begin{align*}
		\Pr[t_1 \geq  64n^3\eta^{-2}(1+\eta)^2\varepsilon^{-1}] =	\Pr\left[t_1  \geq  2\expect[t_1] \right] & \leq \Pr\left[\sum_{i\in [n]} x_i  \geq 2\expect\left[ \sum_{i \in [n]} x_i\right] \right] \\
		&\leq  \sum_{i\in[n]}	\Pr\left[ x_i \geq  2\expect[x_i]  \right] \\ 
		\text{(Prop.~\ref{prop:negbinconc})} \qquad	&\leq \sum_{i\in[n]} \exp(-2k(1-1/2)^2/2)	 = n\exp(-k/4)  \\
		\text{(Substituting $k$ and $\eta \leq 1/5, \varepsilon < 1$)} \qquad&\leq n\exp(-n\eta^{-2}(1+\eta)^2\varepsilon^{-1}) \leq  n\exp(-9n)\leq  1/24
	\end{align*}
	In the last inequality we used the fact that for $s>0$, $xe^{-sx} \leq 1/es$.
\end{proof}

\end{document}